\documentclass[A4]{article}
\usepackage{graphicx,amssymb,amsmath,verbatim,subfigure}
\usepackage{fullpage}
\usepackage{amsthm}
\newtheorem{theorem}{Theorem}

\newtheorem{cor}[theorem]{Corollary}
\newtheorem{obs}[theorem]{Observation}
\newtheorem{prop}[theorem]{Proposition}
\newtheorem{conj}[theorem]{Conjecture}

\newcommand{\union}{\ensuremath\cup}

\newcommand{\figref}[1]{\figurename~\ref{#1}}
\newcommand{\figurenames}{{\figurename}s}
\newcommand{\figrefs}[2]{\figurenames~\ref{#1} and~\ref{#2}}

\newcommand{\mycase}[1]{\vspace{-1ex}\paragraph{\it Case~#1}\vspace{-1ex}}
\newcommand{\mycasestart}{}
\newcommand{\mycaseend}{}

\usepackage[normalem]{ulem} %
\usepackage[usenames,dvipsnames]{xcolor}
\usepackage{paralist} %

\graphicspath{{figures/}}

\title{Empty triangles in good drawings of the complete graph%
}

\author{Oswin~Aichholzer\thanks{Institute for Software Technology,
        University of Technology, Graz, Austria, %
        {\tt [oaich|thackl|apilz|bvogt]@ist.tugraz.at}}
		\and
		Thomas Hackl$^\ast$%
		\and
		Alexander Pilz$^\ast$%
		\and
		Pedro~A.~Ramos\thanks{Departamento de Matem\'aticas, Universidad de Alcal\'a,
        Madrid, Spain, {\tt pedro.ramos@uah.es}}
		\and
		Vera~Sacrist\'an\thanks{Departament de Matem\`atica Aplicada II,
		Universitat Polit\`ecnica de Catalunya, Barcelona, Spain, %
		{\tt vera.sacristan@upc.edu}}
		\and
		Birgit Vogtenhuber$^\ast$%
}

\begin{document}
\maketitle

\begin{abstract}
A good drawing of a simple graph is a drawing on the sphere or, equivalently, in the plane in which vertices are drawn as distinct points,
edges are drawn as Jordan arcs connecting their end vertices,
and any pair of edges intersects at most once.
In any good drawing, the edges of three pairwise connected vertices form a Jordan curve which we call a triangle.
We say that a triangle is empty if one of the two connected components it induces 
does not contain any of the remaining vertices of the drawing of the graph.
We show that the number of empty triangles in any good drawing of the complete graph $K_n$ with $n$ vertices is at least $n$.
\end{abstract}

\section{Introduction} %
Consider a simple graph $G=(V,E)$.
A good drawing $D(G)$ of~$G$ on the sphere $S^2$ or, equivalently, in the Euclidean plane $E^2$ is a drawing with the following properties:
\begin{enumerate}
\item The vertices are drawn as distinct points on the sphere $S^2$ (or in the Euclidean plane $E^2$).
\item The edges are Jordan arcs\footnote{Jordan arcs are non-self-intersecting continuous curves containing their end points.} which have the drawings of the vertices they connect as end points.
\item Edges do not pass through any drawn vertex except for their own end vertices.
\item Any pair of edges intersects in at most one point (either in the interior of both edges, forming a proper crossing; or at a common end point).
\end{enumerate}

Besides being a reasonable restriction for a natural drawing of a graph, a main interest in good drawings comes from the fact that they are useful for minimizing the number of crossings:
It is well-known that if in a drawing of a graph there are edges which have self-intersections or pairs of edges which cross more than once, then the graph can be redrawn with less crossings.
Therefore, only good drawings need to be considered when the goal is to make drawings with few crossings or to minimize the number of crossings. See for example~\cite{RiTh97,PanR07} for results on this topic.

In a good drawing $D(G)$ of a graph $G$, the edges of any three pairwise connected vertices in~$D(G)$ form a Jordan curve\footnote{Jordan curves are continuous non-self-intersecting curves that are closed in the sense that the two ``end points'' are identical.}, which we call a \emph{triangle}. This definition matches the usual definition for the special case of straight-line drawings of~$G$, i.e., drawings in the plane where edges are straight-line segments.
Any triangle, being a Jordan curve, partitions the sphere (or the plane) into two connected components.
If, in~$D(G)$, one of these components does not contain the drawing of any of the remaining vertices, then the triangle is called \emph{empty}.
Further, for the case of a good drawing $D(G)$ in the plane, one of the connected components induced by a triangle is bounded while the other one is unbounded.
We denote the former as interior and the latter as exterior of the triangle.
If, in~$D(G)$, no vertex of~$G$ is drawn in the interior of a triangle, then we denote the triangle as \emph{interior-empty}.
Likewise, if, in~$D(G)$, no vertex of~$G$ is drawn in the exterior of a triangle, then we denote the triangle as \emph{exterior-empty}.%

In this work, we consider the number of empty triangles in good drawings $D(K_n)$ of the complete graph $K_n$ with $n$ vertices.
The question of finding empty triangles in good drawings of the complete graph goes back to Erd\H{o}s' question~\cite{Er78} about the existence
of convex $k$-holes (empty polygons spanned by $k$ vertices and edges) in straight-line drawings of the complete graph $K_n$ and the subsequently posed question about their number~\cite{KaMe88}. %
For the existence question, it is by now well-known that every sufficiently large point set contains empty convex  triangles, quadrilaterals, pentagons~\cite{Ha78}, and also hexagons~\cite{nic,gerk}, but that there exist arbitrarily large point sets without empty convex heptagons~\cite{Ho83}.
While the existence question is trivial for empty triangles, the question on the least number $h_3(n)$ of empty triangles in straight-line drawings of~$K_n$ has attracted many researchers and has been the topic of a large number of publications. The currently best known bounds for $h_3(n)$ are $n^2 - \frac{32}{7}n + \frac{22}{7} \leq h_3(n) \leq 1.6196n^2 + o(n^2)$, where the upper bound is due to B{\'a}r{\'a}ny and Valtr~\cite{BV2004} and the lower bound can be found in~\cite{afhhpv-lbnsc-12-cccg}. Note that both the upper and the lower bound are quadratic in~$n$.

In contrast, for general good drawings, Harborth~\cite{h-etdcg-98} showed in 1989 that it is possible to draw $K_n$ such that it contains only $2n-4$ empty triangles.
Note that this implies that most edges are not incident to any empty triangle, while in straight-line drawings, every edge is incident to at least one empty triangle.
Harborth mentioned in the same work that for $3 \leq n \leq 6$, the number of empty triangles in any good drawing $D(K_n)$ is at least $2n-4$.
For $n \geq 7$, the best general lower bound he could show was~$2$.
However, Harborth conjectured that every vertex in any drawing $D(K_n)$ is incident to at least two empty triangles.
Recently, Fulek and Ruiz-Vargas~\cite{fr-tgetdm-13} proved Harborth's conjecture to be true, thus providing a lower bound of~$\frac{2n}{3}$ for the number of empty triangles in any good drawing $D(K_n)$.
In this paper we improve that bound and show that the number of empty triangles in any such drawing is at least $n$.
Further, for $n \leq 8$, we show that Harborth's upper bound of~$2n-4$ is still tight and we conjecture this to be the case in general.

\paragraph{Outline.} Before proving our main theorem in Section~\ref{sec:main}, we review Ruiz-Vargas' proof in Section~\ref{sec:star} and show that it allows to obtain additional properties of the considered empty triangles.
Further, we investigate the relation between rotation schemes and the task of computing the minimum number of empty triangles in Section~\ref{sec:small} and present results for graphs with few vertices.
In Section~\ref{sec:conclusion}, we conclude by giving a short account on our conjecture that every good drawing contains at least $2n-4$ empty triangles.

Note that for many purposes, including counting empty triangles, drawings on the sphere $S^2$ are equivalent to drawings in the plane $E^2$ by Riemann stereographic projection\footnote{Riemann stereographic projection is a projection from the plane to a tangent sphere (or back) where the projection center lies on the sphere and opposite to the tangent point of the plane.}:
In any good drawing $D(K_n)$ of the complete graph $K_n$, let a \emph{cell (of~$D(K_n)$)} be an open region (of~$S^2$ or $E^2$, respectively) whose boundary is defined by (parts of) drawn edges of~$K_n$ and which does not contain any part of~$D(K_n)$ (i.e., no part of a drawn edge or vertex of~$K_n$).
Then for drawings in~$E^2$ exactly one cell is unbounded, while for drawings on $S^2$ all cells are bounded.
Now consider a drawing $D(K_n)$ on $S^2$ and an arbitrary cell $C$ of~$D(K_n)$.
Applying Riemann stereographic projection with the projection center in~$C$, one obtains a drawing $D'(K_n)$ in~$E^2$ where the unbounded cell is the projection of~$C$. Note that for every triangle~$\Delta$ in~$D(K_n)$, $C$ is completely contained in one of the two connected components of~$S^2$ induced by $\Delta$. Further, note that the projection does not change any crossing properties of the edges.
Thus, all vertices of~$K_n$ which are drawn in the connected component of~$S^2$ induced by $\Delta$ that contains $C$ lie in the exterior of the projection $\Delta'$ of~$\Delta$, while all  vertices of~$K_n$ which are drawn in the other connected component of~$S^2$ induced by $\Delta$ lie in the interior of~$\Delta'$. Particularly, $\Delta$ is empty if and only if $\Delta'$ is (interior- or exterior-)empty.
While both models are equivalent in that sense, in some parts of our reasoning it will be more convenient to consider the drawings in the
plane rather than on the sphere. Especially, all the drawings in all figures are assumed to be in the plane.

\section{Empty star triangles} %
\label{sec:star}

Recall that in a good drawing $D(G)$ of a graph $G$, the edges incident to a vertex $v$ do only intersect in~$D(v)$  (where $D(v)$ is the drawing of~$v$).
Thus, the (drawing of the) graph consisting of all vertices of~$G$ and all edges incident to a vertex $v$ of~$G$ is always crossing-free.
We denote this graph as the \emph{induced star graph (of~$v$ in~$D(G)$)}.
Note that $D(G)$ induces a circular order of the edges incident to $v$; see~\figref{fig:star_gen}.

\begin{figure}[htb]
\centering
\includegraphics[page=4]{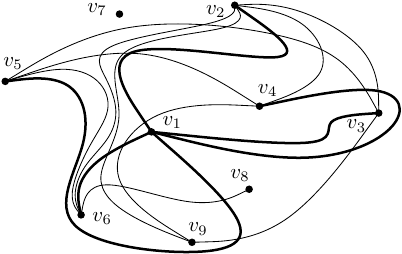}
\caption{\label{fig:star_gen}
	The induced star graph of vertex $v_1$ (drawn bold) in a good drawing of a graph.
	The vertices that are incident to $v_1$ are labeled with respect to their circular order around $v_1$ in~$D(G)$.
}
\end{figure}

If for a triangle $\Delta = D(uvw)$, the (drawing $D(uv)$ of the) edge $uw$ is not crossed by any (drawing of an) edge incident to $v$ in~$D(G)$, then we say that $\Delta$ is a \emph{star triangle (at $v$ in~$D(G)$)}.
In the drawing in~\figref{fig:star_gen}, $D(v_1v_2v_4)$ is a star triangle at $v_1$.
For comparison, $D(v_1v_2v_3)$ is not a star triangle at $v_1$, as the edge $v_2v_3$ crosses the edge $v_1v_4$ in~$D(G)$.
$D(v_1v_5v_6)$ and $D(v_1v_6v_2)$ are other star triangles at $v_1$.
As can be seen in~\figref{fig:star_gen}, the induced star graph of a vertex in a general graph might have isolated vertices.
In contrast, the star graph of a vertex in the complete graph $K_n$ always contains $n-1$ edges and connects all vertices of~$K_n$, see~\figref{fig:star_comp}.

\begin{figure}[htb]
\centering
\subfigure[\label{fig:star_comp1}Induced star graph of~$v_1$ in~$D(K_6)$.
]{\includegraphics[page=5]{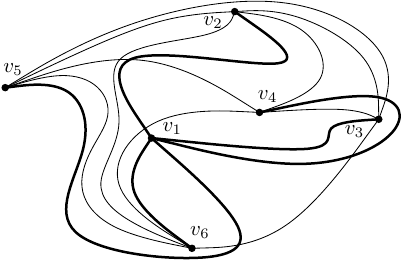}} \hspace{1cm}
\subfigure[\label{fig:star_comp3}Induced star graph of~$v_3$ in~$D(K_6)$.
]{\includegraphics[page=6]{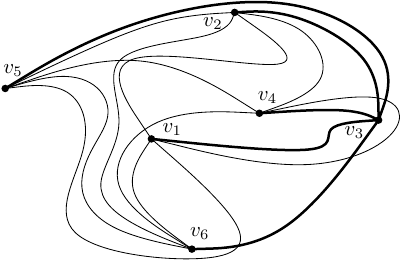}}
\caption{\label{fig:star_comp}
	Examples of induced star graphs (drawn bold) and star triangles in a good drawing $D(K_6)$ of the complete graph $K_6$.
	Vertices are labeled with respect to their circular order around $v_1$.
	$D(v_1v_2v_4)$ is a star triangle at $v_1$.
	$D(v_1v_2v_3)$ is not a star triangle at $v_1$ (as the edge $v_2v_3$ crosses the edge $v_1v_4$ in~$D(G)$), but it is a star triangle at $v_3$.
	$D(v_1v_5v_6)$ is an interior-empty star triangle at $v_1$.
	$D(v_1v_3v_5)$ is an exterior-empty triangle. Further, it is a star triangle at $v_5$, but not at $v_1$ or $v_3$.
}
\end{figure}

This property can be used to obtain the following proposition about star triangles in good drawings of complete graphs.

\begin{prop}
Consider a good drawing $D(K_n)$ of the complete graph $K_n$ with $n \geq 3$ and let $\Delta = D(uvw)$ be a star triangle at $v$.
Then $\Delta$ is empty if and only if $u$ and $w$ are adjacent in the circular order of the edges around $v$ (in~$D(G)$).
\end{prop}
\begin{proof}
Let $H \subseteq K_n$ be the induced star graph of~$v$ plus the edge $uw$.
As $\Delta$ is a star triangle at $v$, $D(H)$ is crossing-free.
Further, $uvw$ is the only simple cycle in~$H$.
Let $V_1$ and $V_2$ be the subsets of vertices of~$K_n$ which are drawn in the two connected components induced by $\Delta$, respectively. %
As $\Delta$ is a star triangle at $v$, $uw$ is not crossed by any of the edges of the induced star graph of~$v$, where the latter contains an edge between $v$ and each other vertex of~$K_n$.
Thus, all edges from $v$ to vertices of~$V_1$ are drawn completely in one connected component induced by $\Delta$, %
and all edges from $v$ to vertices of~$V_2$ are drawn completely in the other connected component induced by $\Delta$, %
implying that the circular order of the vertices around $v$ is $u,V_1,w,V_2$.
Hence, $u$ and $w$ are adjacent in this order if and only if $V_1 = \emptyset$ or $V_2 = \emptyset$, which is equivalent to $\Delta$ being interior-empty or exterior-empty.
\end{proof}

In~\cite[Proposition~3.1]{fr-tgetdm-13}, Fulek and Ruiz-Vargas show that in a good drawing in the plane, every vertex is incident to at least one interior-empty triangle.
He does so by explicitly finding such a triangle $\Delta$.
We reconsider the proof of this proposition, showing that $\Delta$ is in fact a star triangle at $v$.

\begin{prop}\label{prop:1star}
For every good drawing $D(K_n)$ of the complete graph $K_n$ in the Euclidean plane with $n \geq 4$ vertices and every vertex $v$ of~$K_n$, there exists at least one interior-empty star triangle at $v$ in~$D(K_n)$.
\end{prop}
\begin{proof}
Let $H_0$ be the star graph of~$v$, and let $u_0$ be a vertex of~$K_n\backslash\{v\}$; see~\figref{fig:1star_prop} for an accompanying example.
By~\cite[Corollary~2.3]{fr-tgetdm-13}, there is an edge $u_0w_0$, with $w_0 \in K_n\backslash\{u_0,\ v\}$, such that $D(H_0 \union \{u_0w_0\})$ is still a crossing-free drawing.
Consider the triangle $\Delta_0 = D(vu_0w_0)$.
If $\Delta_0$ is interior-empty then it only remains to show that $\Delta_0$ is a star triangle at $v$; see below.
Otherwise, let $H_1 = H_0 \union \{u_0w_0\}$, and let $u_1$ be a vertex of~$K_n$ whose drawing lies in the interior of~$\Delta_0$.
Repeating the argumentation, there is a vertex $w_1 \in K_n\backslash\{u_1,\ v\}$ such that $D(H_1 \union \{u_1w_1\})$ is still crossing-free (and thus, $D(w_1)$ lies in the interior or on the boundary of~$\Delta_0$).
As $\Delta_1 = D(vu_1w_1)$ contains strictly less vertices of~$K_n$ than $\Delta_0$, repeating this process terminates with an interior-empty triangle $\Delta_i=vu_iw_i$ in a crossing-free drawing $D(H_i\union \{u_iw_i\}) \subset D(K_n)$.

\begin{figure}[htb]
\centering
\subfigure[\label{fig:1star_prop_0}%
]{\includegraphics[page=7]{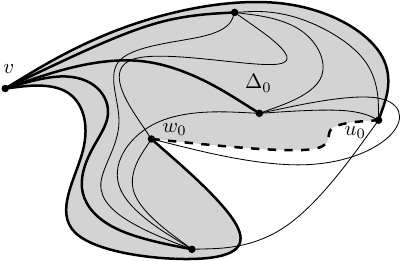}} \hspace{1cm}
\subfigure[\label{fig:1star_prop_1}%
]{\includegraphics[page=8]{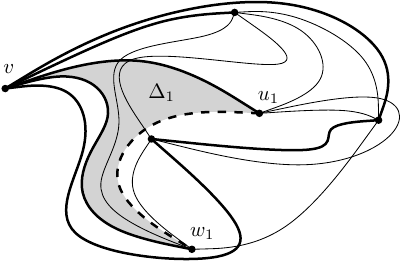}}
\caption{\label{fig:1star_prop}
	Finding interior-empty star-triangles at $v$: $H_i$ is drawn bold, the edge $u_iw_i$ is drawn dashed, and $\Delta_i$ is drawn shaded.
	(a) First step: $\Delta_0$ is not interior-empty.
	(b) Second step: $\Delta_1$ is interior-empty, so this is also the last step in this example.
}
\end{figure}

Finally consider the interior-empty triangle $\Delta_i$, $i\geq 0$, that has been found by this procedure.
As $D(H_i\union \{u_iw_i\})$ contains the star graph of~$v$ in~$D(K_n)$ and is crossing-free, $\Delta_i$ is a star triangle at $v$.
\end{proof}

Consider a good drawing $D(K_n)$ of the complete graph $K_n$ in the Euclidean plane ($n \geq 4$) and let $v$ be a vertex of~$K_n$.
By Proposition~\ref{prop:1star}, there exists at least one interior-empty star triangle $\Delta$ at $v$.
Let $C$ be a cell of~$D(K_n)$ which lies completely in the interior of~$\Delta$.
From  $D(K_n)$, we obtain a good drawing $D'(K_n)$ in the plane where the projection of~$\Delta$ is an exterior-empty triangle by applying Riemann stereographic projection twice:
First project $D(K_n)$ to the sphere. Then project the result back to the plane with the (new) projection center inside the projection of~$C$, i.e., in~$D'(K_n)$ $C$ is the unbounded cell.
Repeating the above proof to the drawing $D'(K_n)$, we obtain an interior-empty triangle $\Delta'$ which is a star-triangle at $v$ in~$D'(K_n)$.
As the %
projection does not change any crossing properties of the edges, the inverse projection of~$\Delta'$ is a star-triangle of~$v$ in~$D(K_n)$ as well (either interior- or exterior-empty).
Similarly, if we have a good drawing on the sphere, we can first project it to the plane (by this making an arbitrary cell unbounded) and then apply the same arguments as above.
Thus, we altogether obtain the following corollary.

\begin{cor}\label{cor:2star}
For every good drawing $D(K_n)$ of the complete graph $K_n$ with $n \geq 4$ vertices and every vertex $v$ of~$K_n$, there are at least two empty star triangles at $v$ in~$D(K_n)$.
\end{cor}

The bounds from Proposition~\ref{prop:1star} and Corollary~\ref{cor:2star} are tight in the sense that there exist drawings of~$K_n$  in the plane where most vertices are incident to exactly one interior-empty and one exterior-empty triangle, or to exactly two interior-empty and no exterior-empty triangles.
See for example \figref{fig:monotone1} for the former and \figref{fig:monotone2} or Harborth's upper bound drawing~\cite[Fig.~1]{h-etdcg-98} for the latter.

\begin{figure}[htb]
\centering
\subfigure[\label{fig:monotone1}%
]{\includegraphics[scale = 0.94, page=4]{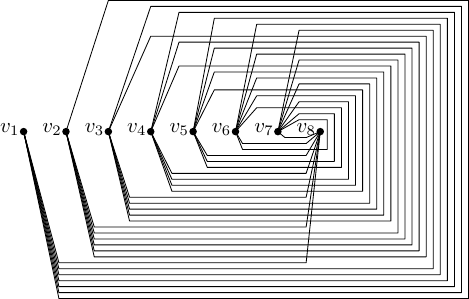}} \hfill
\subfigure[\label{fig:monotone2}%
]{\includegraphics[scale = 0.94, page=5]{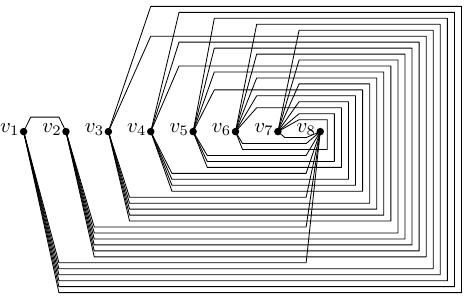}}
\caption{\label{fig:monotone}
	Good drawings of~$K_8$ with $2n-4 = 12$ empty triangles where
	(a) vertices $v_3, \ldots, v_5$ are incident to exactly one interior-empty and one exterior-empty triangle, and
	(b) vertices $v_3, \ldots, v_5$ are incident to exactly two interior-empty and no exterior-empty triangles.
}
\end{figure}

\section{Lonely and lucky vertices} %
\label{sec:main}
Consider a good drawing $D(K_n)$ of the complete graph $K_n$ and a vertex $v$ in this drawing.
If there exists a triangle $\Delta$ in~$D(K_n)$ for which $v$ is the only vertex drawn in one of the two connected components induced by $\Delta$, then we say that $v$ is \emph{lonely in~$\Delta$}.
For example, in the drawing in~\figref{fig:star_comp}, vertex $v_3$ is lonely in~$D(v_1v_2v_4)$ as it is the only vertex drawn in the interior of~$D(v_1v_2v_4)$. Likewise, vertex $v_5$ is lonely in~$D(v_2v_3v_6)$ as it is the only vertex drawn in the exterior of~$D(v_2v_3v_6)$.

\begin{prop}\label{prop:lonely3}
If a vertex $v$ of~$K_n$, $n\geq 4$, is lonely in a good drawing $D(K_n)$, then $v$ is incident to at least three empty triangles in~$D(K_n)$.
\end{prop}
\begin{proof}
Consider a triangle $\Delta = D(v_1v_2v_3)$ in~$D(K_n)$ which witnesses the loneliness of~$v$, i.e., $v$~is the only vertex in one of the two connected components induced by $\Delta$. %
Further, consider the edges $e_1=vv_1$, $e_2=vv_2$, and $e_3=vv_3$ between $v$ and the three vertices of~$\Delta$.
As $D(K_n)$ is a good drawing, at most one of~$e_1$, $e_2$, and $e_3$ can form a crossing with an edge of~$\Delta$. %
We distinguish two cases.
\mycasestart

\begin{figure}[htb]
\centering
\subfigure[\label{fig:lonely_noncrossing_int}%
]{\includegraphics[page=1]{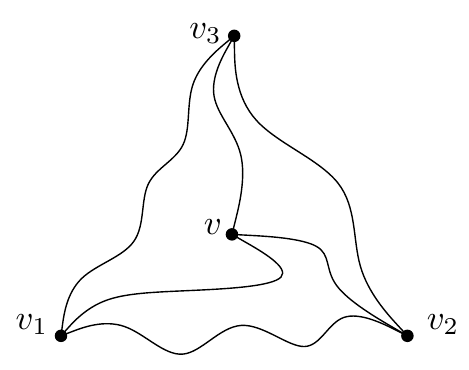}} \hspace{1cm}
\subfigure[\label{fig:lonely_noncrossing_ext}%
]{\includegraphics[page=2]{lonely3}}
\caption{\label{fig:lonely_noncrossing}
	None of the edges between $v$ and $\{v_1,v_2,v_3\}$ crosses an edge of~$\Delta = D(v_1v_2v_3)$:
	(a) $v$ in the interior of~$\Delta$, and
	(b) $v$ in the exterior of~$\Delta$.
}
\end{figure}

\mycase{1: None of~$e_1$, $e_2$, and $e_3$ forms a crossing with an edge of~$\Delta$.}
Note that in this case, $e_1$, $e_2$, $e_3$, and $v$ are all completely in the same connected component induced by $\Delta$. %
Moreover, as none of the other vertices is on this side of~$\Delta$, each of the three triangles formed by $v$ and two vertices of~$\Delta$ is empty; see~\figref{fig:lonely_noncrossing}.

\begin{figure}[htb]
\centering
\subfigure[\label{fig:lonely_crossing_int}%
]{\includegraphics[page=3]{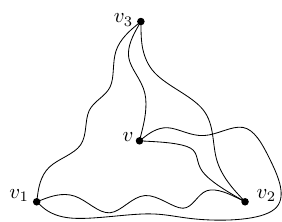}} \hfill
\subfigure[\label{fig:lonely_crossing_ext1}%
]{\includegraphics[page=4]{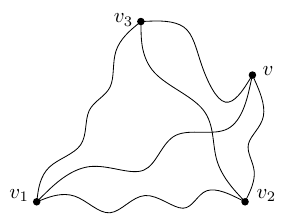}} \hfill
\subfigure[\label{fig:lonely_crossing_ext2}%
]{\includegraphics[page=5]{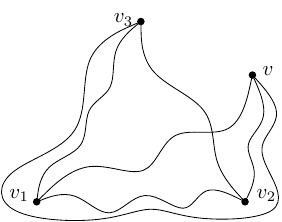}}
\caption{\label{fig:lonely_crossing}
	One of the edges between $v$ and $\{v_1,v_2,v_3\}$, w.l.o.g.\! $vv_1$, crosses an edge of~$\Delta = D(v_1v_2v_3)$:
	(a)~$v$~in the interior of~$\Delta$;
	(b)~$v$~in the exterior of~$\Delta$, $vv_3$ drawn such that $D(vv_2v_3)$ is interior-empty; and
	(c)~$v$~in the exterior of~$\Delta$, $vv_3$ drawn such that $D(vv_2v_3)$ is exterior-empty.
}
\end{figure}

\mycase{2: One of the edges $e_1$, $e_2$, or $e_3$ forms a crossing with an edge of~$\Delta$.}
W.l.o.g., let $e_1=vv_1$ be this edge.
Then the crossed edge of~$\Delta$ is $v_2v_3$, and $\Delta'=D(vv_2v_3)$ is an empty triangle; see~\figref{fig:lonely_crossing}.
As the edge $v_2v_3$ is crossed by $vv_1$, $\Delta'$ is none of the empty star triangles at $v$ which are encountered by the proof of Proposition~\ref{prop:1star} and by Corollary~\ref{cor:2star}.
Thus, together with these two empty star triangles, $v$ is incident to at least three different empty triangles.
\mycaseend
\end{proof}

In the following, let $l(v)$ be the number of triangles in~$D(K_n)$ in which $v$ is lonely, and $t(v)$ be the number of empty triangles in~$D(K_n)$ incident to $v$.
If $t(v)-l(v) \geq 2$ then we say that $v$ is \emph{lucky}.
Note that for $n \geq 4$, every vertex $v$ that is not lonely is lucky, as we know by Corollary~\ref{cor:2star} that in this case $t(v) \geq 2$. Also, every vertex $v$ with $l(v)=1$ is lucky, as in this case $t(v) \geq 3$ by Proposition~\ref{prop:lonely3}.

\begin{theorem}\label{thm:n}
For $n \geq 4$, the number of empty triangles in any good drawing $D(K_n)$ of the complete graph $K_n$ with $n$ vertices is at least $n$.
\end{theorem}
\begin{proof}
We prove the bound by induction on the number $n$ of vertices.
For the induction base, it is straightforward that for $n=4$, every good drawing contains exactly
four empty triangles; see again \figrefs{fig:lonely_noncrossing}{fig:lonely_crossing}.
So assume that the statement is true for any good drawing $D(K_{n'})$ with $n' < n$, and consider a good drawing of~$D(K_{n})$.
We distinguish two cases.

\mycasestart
\mycase{1: $D(K_n)$ contains a vertex $v$ which is lucky.}
As $v$ is lucky, we know that $t(v)-l(v) \geq 2$.
Removing $v$ and all its incident edges results in a drawing $D(K_{n-1})$.
By the induction hypothesis, this drawing contains at least $n-1$ empty triangles.
When adding $v$ and all its incident edges again, the number of empty triangles is increased by $t(v)$ and decreased by $l(v)$.
Thus, $D(K_n)$ contains at least $n-1+t(v)-l(v) \geq n+1 > n$ empty triangles.

\mycase{2: All vertices of~$K_n$ are lonely in~$D(K_n)$.}
By Proposition~\ref{prop:lonely3}, every lonely vertex is incident to $t(v) \geq 3$ empty triangles.
Summing up the number of incident empty triangles per vertex over all vertices,
every triangle is counted exactly thrice (once for each of its vertices).
Thus, $\frac{1}{3} \cdot \sum_{v \in K_n} t(v) \geq \frac{1}{3} \cdot \sum_{v \in K_n} 3 = n$
is a lower bound for the number of empty triangles in~$D(K_n)$.
\end{proof}

\section{Rotation schemes and graphs with few vertices} %
\label{sec:small}
While the emptiness of a triangle clearly depends on the drawing, not all information of the good drawing is needed to decide whether a triangle $v_1 v_2 v_3$ is empty.
While for deciding interior- or exterior-emptiness we need to know which side of the triangle contains the unbounded face, we can decide whether a triangle is empty by only looking at the rotation scheme of the drawing.
Given a drawing of a graph $G$ on an oriented surface, the \emph{rotation scheme} of the drawing of~$G$ gives the circular order of the edges around each vertex of~$G$.
Let $v_1 v_2 v_3$ be a triangle in a good drawing.
The rotation scheme of~$v_2$ is separated by the edges $v_2v_1$ and $v_2v_3$ into two disjoint (possibly empty) sequences.
For any fixed direction of the circular order around $v_2$, let $R_2$ be the sequence between the edges $v_2v_1$ and $v_2v_3$, and let $L_2$ be the sequence between the edges $v_2v_3$ and $v_2v_1$.
For $v_1$ and $v_3$, we define $R_1$, $L_1$ as well as $R_3$ and $L_3$ analogously.
We call a sequence $R_i$ a \emph{right} sequence and a sequence $L_i$ a \emph{left} sequence, for $1\leq i \leq 3$.
In any rotation scheme (and any good drawing) of the complete graph, edges from $v_1$, $v_2$, and $v_3$ to any vertex $v$ are trivially contained either in at least two left sequences or at least two right sequences; we then say that $v$ is \emph{left} of~$v_1 v_2 v_3$ or \emph{right} of~$v_1 v_2 v_3$, respectively.
In a good drawing, the triangle is empty if either all other vertices are left of~$v_1 v_2 v_3$ or all other vertices are right of~$v_1 v_2 v_3$; see again \figrefs{fig:lonely_noncrossing}{fig:lonely_crossing}.

Exhaustively generating all possible rotation schemes of~$K_n$ for small $n$ and counting the number of empty triangles therein can therefore easily be done.
It remains to verify whether a rotation scheme is actually realizable, i.e., is the rotation scheme of at least one good drawing.
Deciding realizability can be done in a combinatorial way by considering a drawing as a crossing-free graph where
\begin{inparaenum}[(i)]
\item each vertex is either a vertex of the original graph or a crossing of the original graph,
\item each edge is a part of an edge of the original graph, and
\item each face is a cell of the original graph.
\end{inparaenum}
For small point sets, a simple backtracking procedure that subsequently adds edges of the original graph and checks whether the drawing is good is sufficient and can be implemented in a straight-forward way.
Note that Kyn\v{c}l~\cite{kyncl} gives a more sophisticated, polynomial-time algorithm to decide realizability of a given rotation scheme of the complete graph.
For $3 \leq n \leq 6$, Harborth mentioned in~\cite{h-etdcg-98} that the number of empty triangles in any good
drawing $D(K_n)$ is at least $2n-4$.
By extensive computer search, we have been able to confirm this result and show the same to be true also for $n = 7$ and $n = 8$.

\begin{obs}
For $3 \leq n \leq 8$, the number of empty triangles in a good drawing of~$K_n$ is at least $2n-4$.
\end{obs}

If every drawing with few empty triangles would contain a lucky vertex, then, by the proof of Theorem~\ref{thm:n}, the number of empty triangles would always be at least $2n-4$.
This is the case for the upper bound example from Harborth~\cite[Fig.~1]{h-etdcg-98}, as well as for the drawings shown in \figref{fig:monotone};
there, none of the vertices $v_3, \ldots, v_{n-2}$ is lonely in any triangle, and thus all of them are lucky.
Unfortunately, the drawing in \figref{fig:nolucky} shows that, in general, an argumentation like this one is not possible.

\begin{figure}[htb]
\centering
\includegraphics[page=1]{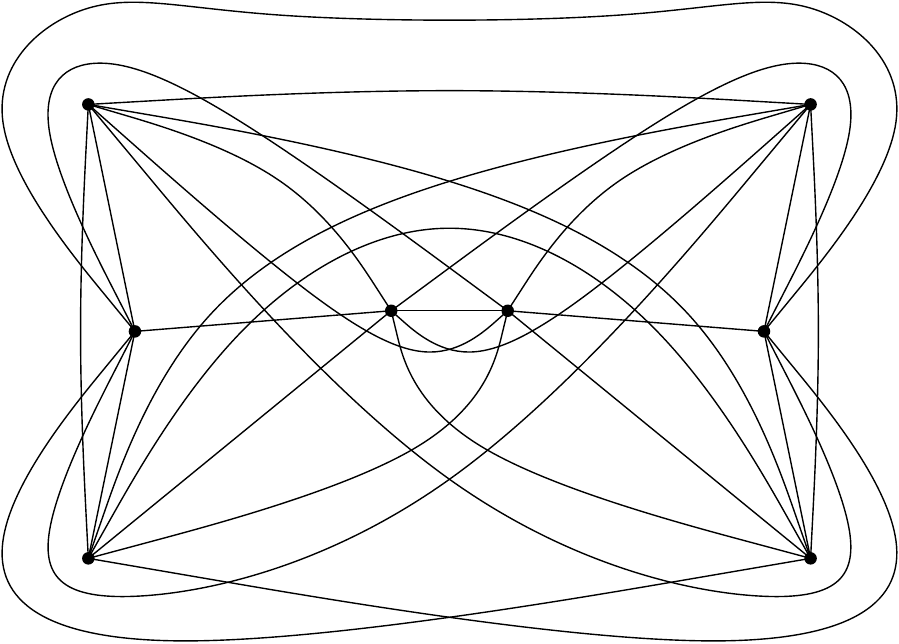}
\caption{\label{fig:nolucky}
	Example of a drawing of the only rotation scheme for $K_8$ where every vertex is at least 2-lonely and no vertex is lucky.
}
\end{figure}

Still, also in this drawing the total number of empty triangles equals $2n-4$.
Moreover,  for $n=8$, the drawing represents the only realizable rotation scheme (out of~$5 370 725$ %
different ones) for which there is no lucky vertex.
All other realizable rotation schemes of cardinality $4 \leq n \leq 8$, no matter whether or not they have few empty triangles, contain at least one lucky vertex.

\section{Conclusion}
\label{sec:conclusion}

In this paper we have shown that any good drawing of the complete
graph $K_n$ with $n$ vertices contains at least $n$ empty triangles,
thus improving the best previous lower bound of~$\frac{2n}{3}$. As
mentioned in the last section, Harborth already stated that the number
of empty triangles for the best known minimizing examples is $2n-4$,
and we have confirmed this for $n$ up to 8. We thus state the
following conjecture.

\begin{conj}
For $n \geq 4$, the number of empty triangles in any good drawing of the complete graph $K_n$ is at least $2n-4$.
\end{conj}

A triangulation is a maximal, crossing-free drawing of a graph such that every face is an empty triangle. In
the plane, the outer face can be an exception, i.e., it might be a
larger face.  It is interesting to observe that any triangulation of
$n$ points on the sphere has $2n-4$ triangular faces. Equivalently,
any triangulation of a set of~$n$ points in the plane with triangular
convex hull consists of~$2n-5$ triangles plus the outer, triangular
face.

A \emph{geometric} graph consists of vertices which are embedded as
points in the plane, and edges which are straight line segments connecting two
such points. It is easy to see that any complete geometric graph
contains a maximal crossing-free sub-graph, that is, a
triangulation. In contrast, it is NP-complete to decide whether a
general (non-complete) geometric graph contains a triangulation as a
sub-graph~\cite{lloyd77}.

Note that in the non-geometric case a good drawing might contain
$2n-4$ empty triangles, but, as these triangles might overlap, no
triangulation as a sub-drawing. See for
example~\figref{fig:lonely_crossing_int} where in any crossing-free
sub-drawing one of the faces has to be at least a quadrilateral. Even if
we allow the outer face of a triangulation in the plane to be larger,
there exist good drawings which do not contain such a triangulation as
a sub-drawing. Thus we raise the following question: What is the
complexity of deciding whether or not a good drawing $D(K_n)$ contains
a triangulation as a sub-drawing?

\subsection*{Acknowledgements.}
This work was initiated during a research visit of Pedro Ramos and Vera Sacrist\'an in May 2013 in Graz, Austria.
Research of Oswin Aichholzer is partially supported by the ESF EUROCORES programme EuroGIGA---CRP `ComPoSe', Austrian Science Fund (FWF): I648-N18.
Research of Thomas Hackl is supported by the Austrian Science Fund (FWF): P23629-N18 `Combinatorial Problems on Geometric Graphs'.
Alexander Pilz is a recipient of a DOC-fellowship of the Austrian Academy of Sciences at the Institute for Software Technology, Graz University of Technology, Austria.
Research of Pedro Ramos is partially supported by MEC grant MTM2011-22792 and by the ESF EUROCORES programme EuroGIGA, CRP ComPoSe, under grant EUI-EURC-2011-4306, for Spain.
Vera Sacrist\'{a}n is partially supported by projects MTM2012-30951, Gen.\ Cat.\ DGR 2009SGR1040, and ESF EUROCORES programme EuroGIGA, CRP ComPoSe, under grant EUI-EURC-2011-4306, for Spain.

{
\bibliographystyle{abbrv}
\bibliography{holesbib}
}

\end{document}